\documentclass[11pt,twoside]{article}

\usepackage{a4}

\usepackage{amssymb,amsmath,amsthm,latexsym}
\usepackage{amsfonts}
  \usepackage{amsfonts}
\usepackage{graphicx}

\usepackage{mathtools,amsthm}
\newtheoremstyle{case}{}{}{}{}{}{:}{ }{}

\usepackage{hyperref}
\usepackage{amsmath, amsfonts}
\usepackage{amssymb, graphicx}
\usepackage{amscd}
\usepackage{textcomp}
\usepackage{palatino}
\usepackage{mathtools}
\usepackage{enumitem}
\usepackage[dvipsnames]{xcolor}

\newtheorem{theorem}{Theorem}[section]

\newtheorem{corollary}[theorem] {Corollary}
\newtheorem{definition}[theorem]{Definition}
\newtheorem{example}[theorem]{Example}
\newtheorem{lemma} [theorem]{Lemma}

\newtheorem{proposition}[theorem]{Proposition}
\newtheorem{remark}[theorem]{Remark}

\numberwithin{subcase}{case}

\vspace{5cm}

\begin{document}
  
  \label{'ubf'}  
\setcounter{page}{1}                                 
\markboth {\hspace*{-9mm} \centerline{\footnotesize \sc
Cyclic MDS and non-MDS matrices }
               }
              { \centerline                           {\footnotesize \sc  
 T. Chatterjee and A. Laha
 } \hspace*{-9mm}              
    }

\vspace*{-2cm}

\begin{center}
{ 
       { \textbf { A note on cyclic MDS and non-MDS matrices
                               }
       }
\\

\medskip
{\sc Tapas Chatterjee }\\
{\footnotesize Indian Institute of Technology Ropar, Punjab, India.
}\\
{\footnotesize e-mail: {\it tapasc@iitrpr.ac.in}}
\medskip

{\sc Ayantika Laha }\\
{\footnotesize Indian Institute of Technology Ropar, Punjab, India.
}\\
{\footnotesize e-mail: {\it 2018maz0008@iitrpr.ac.in}}
\medskip
}
\end{center}

\thispagestyle{empty} 
\vspace{-.4cm}

\hrulefill
\begin{abstract}
In $1998,$ Daemen {\it{ et al.}} introduced a circulant Maximum Distance Separable (MDS) matrix in the diffusion layer of the Rijndael block cipher,  drawing significant attention to circulant MDS matrices. This block cipher is now universally acclaimed as the AES block cipher. In $2016,$ Liu and Sim introduced cyclic matrices by modifying the permutation of circulant matrices and established the existence of MDS property for orthogonal left-circulant matrices, a notable subclass within cyclic matrices. While circulant matrices have been well-studied in the literature, the properties of cyclic matrices are not. Back in $1961$, Friedman introduced $g$-circulant matrices which form a subclass of cyclic matrices. In this article, we first establish a permutation equivalence between a cyclic matrix and a circulant matrix.
We explore properties of cyclic matrices similar to $g$-circulant matrices. Additionally, we determine the determinant of $g$-circulant matrices of order $2^d \times 2^d$ and prove that they cannot be simultaneously orthogonal and MDS over a finite field of characteristic $2$. Furthermore, we prove that this result holds for any cyclic matrix. 
\end{abstract}
\hrulefill

{\textbf{Keywords}: Circulant Matrices, Cyclic Matrices, g-Circulant Matrices,  Orthogonal Matrices, Permutation Matrices, MDS Matrices.}

{\small \textbf{2020 Mathematics Subject Classification.} Primary: 94A60, 12E20, 15B10 ; Secondary: 15B05 .}\\

\vspace{-.37cm}

\section{\bf Introduction}

Symmetric key primitives, like block ciphers and hash functions, rely on various components that provide ``confusion and diffusion". These concepts were introduced by Claude Shannon in $1949$, in his seminal paper ``Communication Theory of Secrecy Systems" \cite{CS}. Confusion and diffusion layers play a pivotal role in the security and efficiency of cryptographic schemes based on block ciphers and hash functions. The goal of the confusion layer is to hide the statistical relationship between the ciphertext and message whereas, the diffusion layer guarantees that every bit of the message and secret key affects multiple bits of the ciphertext, ensuring that, after several rounds, all output bits depend on every input bit. Therefore, the choice of a diffusion layer plays an important role in providing security against differential cryptanalysis \cite{BS} and linear cryptanalysis \cite{M} which are the two most important cryptanalysis of block ciphers. The strength of a diffusion layer is usually assessed by its branch number. As a result, constructing diffusion layers with higher branch numbers and low-cost implementations is a challenge for cipher design.  Consequently,  many block ciphers use Maximum Distance Separable (MDS) matrices, known for their optimal branch numbers, in their diffusion layers. These matrices are redundant parts of systematic generator matrices of MDS codes and used in popular ciphers like AES \cite{DR}, LED \cite{GPPR}, SQUARE \cite{DKR}, SHARK \cite{RDPBW} as well as hash functions such as Whirlpool \cite{BR} and PHOTON \cite{GPP}.

The construction of diffusion layer using MDS matrices  primarily employs two main approaches: recursive and non-recursive \cite{GPRS}. In recursive construction, a companion or sparse matrix $A$ of order $n \times n$ is used such that $A^n$ is an MDS matrix. Block cipher LED \cite{GPPR}, and hash function PHOTON \cite{GPP} used this recursive based construction. In contrast, non-recursive methods encompass various techniques, such as the use of Cauchy matrices \cite{GR1, MS}, Vandermonde matrices \cite{LF, SDMO}, Hadamard matrices \cite{SKOP}, and their generalizations \cite{PSADR}. Furthermore, inspired by the use of circulant MDS matrix in the diffusion layer of AES, various authors \cite{CL,GR2,GR,LS} have explored the construction of circulant MDS matrices. In \cite{GR}, the authors extended circulant matrices to circulant-like and generalized circulant MDS matrices.

Another interesting property an MDS matrix should possess is being involutory or orthogonal. This is important because the decryption process requires the use of an inverse matrix. In their works \cite{GR1, GR2, GR}, Gupta and Ray studied the existence of circulant MDS matrices with involutory and orthogonal properties. They proved the non-existence of $2^d \times 2^d, d \geq 2$ circulant orthogonal matrices with MDS property over the finite field $\mathbb{F}_{2^m}$. In \cite{GR}, Gupta {\it{et al.}} introduced Type-I circulant-like matrices and demonstrated that they cannot be MDS and orthogonal simultaneously for even orders. Subsequently, in \cite{GPRS}, it was proven that Type-I circulant-like MDS matrices of odd order also lack orthogonality.

In search of orthogonal MDS matrices, Liu and Sim introduced cyclic matrices as a generalization of circulant matrices \cite{LS}. This broader class of matrices is expected to exhibit orthogonal or involutory properties, which certain orders of circulant matrices lack over a characteristic $2$ field.
While, in \cite{LS}, it was proved that $2^d \times 2^d$ left-circulant MDS matrices, a subclass of cyclic matrices, cannot be orthogonal, the orthogonal property for the larger class remained an open question.
Furthermore, in \cite{CL}, the authors introduced  $\theta$-circulant matrices over the quasi polynomial ring. Later, in $2022$, Adhiguna {\it{et al.}} studied the orthogonal property of $\theta$-circulant MDS matrices in \cite{AAYA} over the quasi polynomial ring. Several other authors \cite{TAS, DCG} extended the search for MDS matrices over rings and modules. The authors of \cite{PT1, PT} studied orthogonal MDS matrices over the Galois ring $GR(2^n,k)$. In $2023,$ Mohsenifar {\it{et al.}} \cite{MS1} studied the connection between cyclic matrices and generalized Cauchy matrices. They proved that it is possible to construct circulant MDS matrices of orders $3 \times 3$ and $ 5 \times 5$ under certain conditions on the entries and the finite fields using generalized Cauchy matrices. Recently, in $2024$, the authors in \cite{MEG} introduced a direct construction of circulant MDS matrices with easily computable inverses using Vandermonde matrix.  They also introduced circulant-like matrices and provided a direct construction method for such matrices.

The structure of circulant matrices has been extensively studied in literature \cite{D}. In $1961$, Friedman introduced $g$-circulant matrices as a generalization of the circulant matrix. After that, several authors \cite{AB, D, F} studied the $g$-circulant matrices and their properties. In \cite{D}, a representation of $g$-circulant matrices using permutation matrices is provided. Note that, in \cite{GR2, GR}, authors investigated MDS circulant matrices using the representation of circulant matrix in terms of permutation matrix. Subsequently, in $2016,$ Liu and Sim proved that it is possible to construct left-circulant involutory matrices with MDS property. Additionally, they defined cyclic matrices and a connection between cyclic and circulant matrices. Using the search method, they also claimed the non-existence of involutory cyclic MDS matrices of order $4$ and $8$. In the forthcoming paper \cite{TA2}, the authors studied $g$-circulant matrices with both involutory and MDS property and in \cite{TA, TA1, TA3} Chatterjee {\it{et al.}} studied circulant matrices with semi-involutory and semi-orthogonal properties.

\section{\bf Contributions}
In this article, we establish a structure of cyclic matrices by exploring the connection between cyclic and circulant matrices. This result can be thought of as a generalization of the representation of $g$-circulant matrices introduced in $1961$. 
In \S~\ref{sec:structure_cyclic}, we prove that there exists a unique permutation matrix that converts a cyclic matrix to a circulant matrix. Additionally, we present a general structure of cyclic matrices using permutation matrices. In \S~\ref{sec:orthogonal_property}, we calculate the determinants of $g$-circulant matrices of order $2^d \times 2^d$ over the finite field $\mathbb{F}_{2^m}$. Using this determinant, we prove that $g$-circulant orthogonal matrices of order $2^d \times 2^d$ cannot be MDS over a finite field of characteristic $2$. Moreover, we extend this result for cyclic orthogonal matrices over the finite field of characteristic $2$.

\section{\bf Notations and Preliminaries}
 In this section, we recall some basic notations and definitions related to MDS matrices, circulant matrices, cyclic matrices, etc., along with several useful results essential for the subsequent sections.
 
Let $\mathbb{F}_q$ denote the finite field with $q=p^m$ elements where $p$ is a prime number and $m$ is a positive integer. Consider an $[n,k,d]$ linear error-correcting code $\mathcal{C}$ over $\mathbb{F}_q$. This code $\mathcal{C}$ is a subspace of $\mathbb{F}_q^n$ of dimension $k$ and the minimum distance between any two vectors in $\mathcal{C}$ is $d$. The generator matrix $G$ of $\mathcal{C}$ is a $k \times n$ matrix such that any set of $k$ columns of $G$ are linearly independent. The standard form of $G$ is $[I|A]$ where $I$  is a $k \times k$ identity matrix and $A$ is $k \times (n-k)$ matrix. The Singleton bound sets an upper bound on the minimum distance of a code with given dimensions. It says, for an $[n,k,d]$ code, $d \leq n-k+1$. A code that meets this bound is known as a maximum distance separable (MDS) code.  An alternative definition of an MDS code, as given in \cite{MS}, is the following.
\begin{definition}
An $[n,k,d]$ code $\mathcal{C}$ with generator matrix $G=[I|A]$, where $A$ is a $k \times (n-k)$ matrix is MDS if and only if every square sub-matrix formed from any $i$ rows and any $i$ columns of $A$, for any $i=\{1,2,\hdots,\text{min}(k,n-k)\}$, is non-singular.
\end{definition}
This definition leads to the following characterization of an MDS matrix.
\begin{definition}
A square matrix $A$ is said to be MDS if every square sub-matrix of $A$ is non-singular.
\end{definition}

Next, we outline several properties of MDS matrices, with comprehensive details provided in \cite{GPRS}. Let $A^{-1}$ represent the inverse of matrix $A$, and $A^{T}$ denote its transpose. Additionally, $I$ signifies the identity matrix. For an $k\times k$ matrix $A$, its rows are denoted as $R_i$, where $0 \leq i \leq k-1$, and its columns as $C_j$, where $0 \leq j \leq k-1$. Moreover, $A(i,j)$ denotes the entry at the intersection of the $i$-th row and $j$-th column.

\begin{proposition}
Let $A$ be an MDS matrix. Then $A^{-1}, ~A^{T}$, and $D_1AD_2$ are also MDS matrices where $D_1$ and $D_2$ are two non-singular diagonal matrices.
\end{proposition}

As discussed earlier, the diffusion power of the diffusion layer is quantified by the
branch number (see \cite{DR}, Chapter $9$) of the diffusion matrix.
\begin{definition}
The differential branch number of a linear transformation $\phi$ over the finite field $\mathbb{F}_{2^m}$ is given by
$$\mathcal{B}_d(\phi)=\underset{\textbf{a} \neq \textbf{0},\textbf{a} \in \mathbb{F}_{2^m}^k}{\min} \{\text{Wt}(\textbf{a})+ \text{Wt}(\phi(\textbf{a}))\}=\underset{\textbf{a} \neq \textbf{0}, \textbf{a} \in \mathbb{F}_{2^m}^k}{\min} \{\text{Wt}(\textbf{a})+ \text{Wt}(M\textbf{a})\},$$ where $\phi(\textbf{x})= M \cdot \textbf{x}$.
\end{definition}
\begin{definition}
The linear branch number of a linear transformation $\phi$ over the finite field $\mathbb{F}_{2^m}$ is given by
$$\mathcal{B}_l(\phi)= \underset{\textbf{a} \neq \textbf{0}, \textbf{a} \in \mathbb{F}_{2^m}^k}{\min} \{\text{Wt}(\textbf{a})+ \text{Wt}(M^t\textbf{a})\},$$ where $\phi(\textbf{x})= M \cdot \textbf{x}$.
\end{definition}
Note that the maximal value of $\mathcal{B}_d(\phi)$ and $\mathcal{B}_l(\phi)$ are $k+1$. In general $\mathcal{B}_d(\phi) \neq \mathcal{B}_l(\phi)$, but if a matrix achieves the maximum possible differential or linear branch number, then both branch numbers are equal. Consequently, for an MDS matrix $M$, $\mathcal{B}_d(\phi) =\mathcal{B}_l(\phi) = k+1$.

A matrix $P$ is a permutation matrix if it is obtained from the identity matrix by permuting the rows (or columns).
\begin{definition}
Two matrices $M$ and $M'$ are said to be permutation equivalent if there exist two permutation matrices $P$ and $Q$ such that
$M'=PMQ$. We denote it by $M\sim_{P.E} M'$.
\end{definition}
For permutation equivalent matrices, we have the following result from \cite{LS}.
\begin{proposition}
For any permutation matrices $P$ and $Q$, the branch numbers of $M$ and $PMQ$ are the same.
\end{proposition}
As a consequence of the above proposition, we have the following corollary. 
\begin{corollary}\label{per eq are MDS}
Let $A$ be an MDS matrix. Then for any permutation matrices $P$ and $Q$, $PAQ$ is an MDS matrix.
\end{corollary}
Involutory and orthogonal matrices play an important role in the design of lightweight block ciphers due to the use of inverse matrix in the decryption layer.
\begin{definition}
A square matrix $A$ is said to be involutory if $A^2=I$ and orthogonal if $AA^T=A^TA=I$.
\end{definition}

Following that, we give the definitions of circulant matrices and their generalizations. 
We commence with the definition of a circulant and left-circulant matrix. 
\begin{definition}
The square matrix of the form $\begin{bmatrix}
c_0 & c_1 & c_2 & \cdots & c_{k-1}\\
c_{k-1} & c_0 & c_1 & \cdots & c_{k-2}\\
\vdots & \vdots & \vdots & \cdots & \vdots\\
c_1 & c_2 & c_3 & \cdots & c_0
\end{bmatrix}$ is said to be circulant matrix and denoted by $\mathcal{C}=$ circulant$(c_0 , c_1 , c_2 , \hdots , c_{k-1})$. On the other hand, the square matrix of the form $\begin{bmatrix}
c_0 & c_1 & c_2 & \cdots & c_{k-1}\\
c_{1} & c_2 & c_3 & \cdots & c_{0}\\
\vdots & \vdots & \vdots & \cdots & \vdots\\
c_{k-1} & c_0 & c_1 & \cdots & c_{k-2}
\end{bmatrix}$ is said to be left-circulant matrix and denoted by left-circulant$(c_0 , c_1 , c_2 , \hdots , c_{k-1}).$
\end{definition}
The entire circulant matrix is clearly defined by its first row (or column).
We can express the $(i,j)$-th entry of $\mathcal{C}$ as $\mathcal{C}(i,j)=c_{j-i},$ where subscripts are calculated modulo $k$. Furthermore, the entries of $\mathcal{C}$ follow the relationship $\mathcal{C}(i,j)=\mathcal{C}(i+1,j+1)$. Using the property of permutation matrix $P=$ circulant$(0,1,0,\hdots,0)$, a circulant matrix $\mathcal{C}$ can be expressed as
\begin{eqnarray}\label{circulant structure}
\mathcal{C}= \text{circulant}(c_0 , c_1 , c_2 , \hdots , c_{k-1})=c_0I+c_1P+c_2P^2+\cdots +c_{k-1}P^{k-1},
\end{eqnarray}
where $I$ denotes the identity matrix of order $k\times k$.
The structure and properties of circulant matrices have been extensively explored in the literature. Building upon this foundation, B. Friedman extended this theory in 1961 by introducing $g$-circulant matrix \cite{F}. The definition of a $g$-circulant matrix is as follows:

\begin{definition}\label{g-circ def}
 A $g$-circulant matrix of order $k \times k$ is a matrix of the form $$A= g-\text{circulant}(c_0,c_1,\hdots,c_{k-1})=\begin{bmatrix}
c_0 & c_1 & \cdots & c_{k-1}\\
c_{k-g} & c_{k-g+1} & \cdots & c_{k-1-g}\\
c_{k-2g} & c_{k-2g+1} & \cdots & c_{k-1-2g}\\
\vdots& \vdots &\cdots &\vdots\\
c_{g} & c_{g+1} & \cdots & c_{g-1}\\
\end{bmatrix},$$ where all subscripts are taken modulo $k$. 
\end{definition}
For $g=1$, it represents a circulant matrix, and for $g\equiv -1 \pmod{k}$, it takes the form of a left-circulant matrix.
Here are some noteworthy properties of $g$-circulant matrices, with details provided in \cite{AB, D}.

\begin{lemma}\label{gh-circulant}
Let $A$ be $g$-circulant and $B$ $h$-circulant. Then $AB$ is $gh$-circulant.
\end{lemma}
\begin{lemma}
If $A$ and $B$ are both $g$-circulant matrices, then $AB^{T}$ forms a circulant matrix.
\end{lemma}
\begin{lemma}\label{PA=AP^g}
$A$ is $g$-circulant if and only if $PA=AP^g$, where $P$ is the permutation matrix $P=$ circulant$(0,1,0,\hdots,0)$.
\end{lemma}
The classification of $g$-circulant matrices is divided into two types depending on $\gcd(k,g)$. For the case $\gcd(k,g)=1$, $g$-circulant matrices represent a subclass of cyclic matrices and otherwise, it is not. The case $\gcd(k,g)>1$ is discussed in Section ~\ref{sec:structure_cyclic}.  The notion of cyclic matrices were introduced by Liu and Sim \cite{LS} as a generalization of circulant matrices in $2016$.
Cyclic matrices of order $k \times k$ are defined using a $k$-cycle permutation $\rho$ of its first row where $\rho  \in S_k$, the permutation group on $k$ elements. For instance, consider a $4$-cycle permutation $\rho=\begin{pmatrix}
0 & 1 & 2 & 3\\
2 &0 &3 & 1
\end{pmatrix}$ in $S_4$. This can be written as $\rho=(0 ~2 ~3~1)$ in the cycle notation and we use this notation throughout the paper. The definition of cyclic matrix is the following:

\begin{definition}\label{cyclic def}
For a $k$-cycle $\rho \in S_k$, a matrix $\mathfrak{C}_\rho$ of order $k \times k$ is called cyclic matrix if each subsequent row is $\rho$-permutation of the previous row. We represent this matrix as $\text{cyclic}_\rho(c_0,c_1,c_2,\hdots,c_{k-1})$, where $(c_0,c_1,c_2,\hdots,c_{k-1})$ is the first row of the matrix. The $(i,j)$-th entry of $\mathfrak{C}_\rho$ can be expressed as $\mathfrak{C}_\rho(i,j)=c_{\rho^{-i}(j)}$.
\end{definition}

For example, the matrix $cyclic_\rho(c_0,c_1,c_2,\hdots,c_{k-1})$, where $\rho=(0 ~1 ~ 2 \cdots~ k-1)$ results in a circulant matrix. Similarly, if we use $\rho =(0 ~ k-1 ~1 ~2 \cdots k-2)$ , we obtain a left-circulant matrix.  Note that, a $k$-cycle of the form $\begin{pmatrix} \label{g-cycle}
 0 & 1 & 2 & \cdots & k-g & k-g+1 & \cdots &k-1\\
 g & g+1 & g+2 & \cdots & g+(k-g) & g+(k-g+1) & \cdots & g+k-1
 \end{pmatrix}$, where $(g+i)$ is calculated modulo $k$ and $\gcd(k,g)=1$ can be written as $(0 \quad g \quad {2g\pmod k} \quad {3g\pmod k} \cdots {(k-1)g \pmod k})$. This gives a complete $k$- cycle because of the next lemma. 
 \begin{lemma}\label{gcd result}
Let $S= \{\alpha g \pmod k, ~\alpha =0,1,\hdots ,k-1\}.$ $S$ will be a complete residue system modulo $k$ if and only if $\gcd(k,g)=1$.
\end{lemma}
\begin{proof}
Consider the mapping $\phi: \mathbb{Z}_k \rightarrow S$  defined as $\phi(\alpha)=\alpha g \pmod k$. Take arbitrary elements $\alpha_1, \alpha_2 \in \mathbb{Z}_k$ and assume that $\phi(\alpha_1)=\phi(\alpha_2)$. This implies $\alpha_1 g \pmod k=\alpha_2 g \pmod k$ and hence $k$ divides $(\alpha_1-\alpha_2)g$. Since $\gcd(g,k)=1$, this implies $\alpha_1 - \alpha_2= 0 \pmod k$, yielding  $\alpha_1 = \alpha_2$. Therefore $\phi$ is injective and a bijection. Hence proved.
\end{proof}

For the permutation $\rho=(0 \quad g \quad {2g\pmod k} \quad {3g\pmod k} \cdots {(k-1)g \pmod k})$, $\rho^{-1}$ is of the form $\rho^{-1}=\begin{pmatrix}
 0 & 1 & 2 & \cdots  &k-1\\
 k-g & k-g+1 & k-g+2 & \cdots & k-1-g
 \end{pmatrix}$. Thus cyclic matrices corresponding to these cycles are $g$-circulant matrices. An example of this is the following:

\begin{example}\label{example 1}
Consider two $5$-cycles $\rho_1=(0~2~3~1~4)$ and $\rho_2=(0~3~1~4~2)$ in $S_5$. Then 

$\text{cyclic}_{\rho_1}(c_0,c_1,c_2,c_3,c_4)=\begin{bmatrix}
c_0&c_1&c_2&c_3&c_4\\
c_4&c_3&c_0&c_2&c_1\\
c_1&c_2&c_4&c_0&c_3\\
c_3&c_0&c_1&c_4&c_2\\
c_2&c_4&c_3&c_1&c_0
\end{bmatrix}$ and 

$\text{cyclic}_{\rho_2}(c_0,c_1,c_2,c_3,c_4)=\begin{bmatrix}
c_0&c_1&c_2&c_3&c_4\\
c_2&c_3&c_4&c_0&c_1\\
c_4&c_0&c_1&c_2&c_3\\
c_1&c_2&c_3&c_4&c_0\\
c_3&c_4&c_0&c_1&c_2
\end{bmatrix}$. 

Observe that $\text{cyclic}_{\rho_2}(c_0,c_1,c_2,c_3,c_4)$ is a $3$-circulant matrix.
\end{example}
In the next section we discuss some results on $g$-circulant matrices and generalize them to cyclic matrices. In addition, we explore some properties of cyclic matrices.

\section{\bf{Structure of cyclic matrices and their connection with g-circulant matrices\label{sec:structure_cyclic}}}

The classification of $g$-circulant matrices divided into two types depending on greatest common divisor of $k$ and $g$. If $\gcd(k,g)=1$, it becomes apparent that $g$- circulant matrices are essentially cyclic matrices corresponding to the $k$-cycle 
$(0~ g ~~ {2g\pmod k} ~~ {3g\pmod k} \cdots {(k-1)g \pmod k})$.

In \cite{AB,D} a detailed study on the properties of $g$-circulant matrices is provided.  
We first show that we are only interested in the case $\gcd (k,g)=1$. This restriction is vital because if $\gcd (k,g) > 1,$ these $g$-circulant matrices can never be MDS. To justify this, we use Lemma \ref{gcd result}.

\begin{theorem}
Let $A$ be a $g$-circulant matrix of order $k \times k$. If $\gcd (k,g)> 1$, then $A$ cannot be an MDS matrix.
\end{theorem}

\begin{proof}
Let $\gcd (k,g)=d$. Then there exists $k_1,k_2$ such that $k=dk_1,g=dk_2$ and $\gcd(k_1,k_2)=1$. In a $g$-circulant matrix, the entries satisfy the relation $a_{i,j}=a_{i+1,j+g}$ for $i,j=0,1,2,\hdots,k-1$, with suffixes calculated modulo $k$. Let us consider the first row of $A$ as $(a_{0,0},a_{0,1},a_{0,2}, \hdots, a_{0,n-1})$.  According to Lemma \ref{gcd result}, there exists an integer $\alpha, 1 < \alpha \leq k-1$ such that $\alpha g=0$. Consequently, we have 
\begin{align*}
&A(0,0)= A(1,g)=A(2,2g)=\cdots=A(\alpha,0)=a_{0,0},~\text{and}\\
&A(0,j)= A(1,g+i)=A(2,2g+j)=\cdots=A(\alpha,\alpha g+j)=a_{0,j},~\text{for all}~j=0,1,\hdots, k-1.
\end{align*}
Therefore, first row and $\alpha$-th row are identical and $A$ is not MDS.
\end{proof} 

The next theorem extends the structure defined in Equation (\ref{circulant structure}) to $g$-circulant matrices.
Let $P=$ circulant$(0,1,0,\hdots,0)$ and $Q_g= g$-circulant$(1,0,0,\hdots,0)$. The representation of $g$-circulant matrices is established by the following theorem [\cite{D}, Theorem $5.1.7$]. 
\begin{theorem}\label{structure g circulant}
Let $A= g$-circulant$(c_0,c_1,\hdots,c_{k-1})$ with $\gcd(k,g)=1$. Then $A$ can be expressed as $A=\sum_{i=0}^{k-1} c_iQ_gP^{i}$, where $P=$ circulant$(0,1,0,\hdots,0)$ and $Q_g= g$-circulant$(1,0,0,\hdots,0)$.
\end{theorem}
One can give the analogue of Theorem \ref{structure g circulant} for cyclic matrices. To do so, we first establish the following theorem and the subsequent lemma.

\begin{theorem}\label{relation cyclic circulant}
Let $\mathfrak{C}_\rho(c_0,c_1,\hdots,c_{k-1})$ be a cyclic matrix.  
Then there exists a unique permutation matrix $Q$ such that $\mathfrak{C}_\rho Q=$ circulant$(c_{0},c_{\rho(0)},c_{\rho^2(0)},c_{\rho^3(0)},\hdots,c_{\rho^{k-1}(0)})$.
\end{theorem}

\begin{proof}
Let $\mathfrak{C}_\rho=$ cyclic$(c_0,c_1, c_2, \cdots, c_{k-1})$. Then by Definition \ref{cyclic def}, we have $$\mathfrak{C}_\rho=(c_{\rho^{-i}(j)})=\begin{bmatrix}
c_0 & c_1 & c_2 & \cdots & c_{k-1}\\
c_{\rho^{-1}(0)} & c_{\rho^{-1}(1)} & c_{\rho^{-1}(2)} & \cdots & c_{\rho^{-1}(k-1)}\\
c_{\rho^{-2}(0)} & c_{\rho^{-2}(1)} & c_{\rho^{-2}(2)} & \cdots & c_{\rho^{-2}(k-1)}\\
\vdots & \vdots & \vdots & \cdots & \vdots\\
c_{\rho^{-(k-1)}(0)} & c_{\rho^{-(k-1)}(1)} & c_{\rho^{-(k-1)}(2)} & \cdots & c_{\rho^{-(k-1)}(k-1)}\\
\end{bmatrix}.$$ 

Let us consider the permutation matrix $Q$ such that its $(i,j)$-th entry is given by the rule $$Q(i,j)=\begin{cases}
  1 & \text{if $i=\rho^j(0),j=0,1,\cdots, k-1;$ }\\
  0 & \text{otherwise}.
\end{cases}$$
Since $\rho$ is a cycle of length $k$, we have $\rho^i(0) \neq \rho^l(0)$ for $i,l=\{0,1,2,\hdots,k-1\}, i \neq l$. Therefore $Q$
is a permutation matrix.

Let $\mathfrak{C}_\rho Q= \mathcal{C}$. We prove that $\mathcal{C}$ is a circulant matrix by showing that $i$-th row of $\mathcal{C}$ is a right shift of $(i-1)$-th row for $i=0,1,\hdots, k-1$, i.e., $\mathcal{C}(i,j)= \mathcal{C}(i-1,j-1)$ for $1 \leq i,j \leq k-1$. We will establish this property through induction on $i$.

For the base case, consider $i=0$.

Evaluating $\mathcal{C}(0,0)$ we get, $\mathcal{C}(0,0)=\sum_{l=0}^{k-1}\mathfrak{C}_\rho(0,l)Q(l,0)=\mathfrak{C}_\rho(0,0)Q(0,0)= c_0$ since $Q(0,0)=1$ and $Q(l,0)=0$ for $l=2,3,\hdots,k-1$.

Calculating similarly the other entries of this row, we get $\mathcal{C}(0,j)=\sum_{l=0}^{k-1}\mathfrak{C}_\rho(0,l)Q(l,j)=\mathfrak{C}_\rho(0,\rho^j(0))Q(\rho^j(0),j)= c_{\rho^j(0)}$. This holds because $Q(\rho^j(0),j)=1$ and the other entries of the $j$-th column of $Q$ are $0$.
Consequently, the first row of the matrix  $\mathcal{C}$ is given by $(c_{0},c_{\rho(0)},c_{\rho^2(0)},c_{\rho^3(0)},\hdots,c_{\rho^{k-1}(0)}).$

Next we prove the statement for $i=1$. 

Evaluating $\mathcal{C}(1,j)$ for $j=0,1,\hdots, k-1$, we get $\mathcal{C}(1,j)=\sum_{l=0}^{k-1}\mathfrak{C}_\rho(1,l)Q(l,j)=\mathfrak{C}_\rho(1,\rho^j(0))Q(\rho^j(0),j)= c_{\rho^{-1}(\rho^j(0))}=c_{\rho^{j-1}(0)}= \mathcal{C}(0,j-1)$. This relationship holds true due to the definitions of both $Q$ and $\mathfrak{C}_\rho$. Therefore the row $R_1$ of $\mathcal{C}$ is $(c_{\rho^{k-1}(0)},c_{0},c_{\rho(0)},c_{\rho^2(0)},c_{\rho^3(1)},\hdots,c_{\rho^{k-2}(0)}).$

With this, the validity of the induction hypothesis for both $i=0,1$ are proved. Let's assume it holds for the $i$-th row which can be represented as $(c_{\rho^{k-i}(0)},c_{\rho^{k-i+1}(0)}, \hdots, c_{\rho^{k-i+j-1}(0)},c_{\rho^{k-i+j}(0)}, \hdots,c_{\rho^{k-i-1}(0)})$.

Calculating the first entry of the $i+1$-th row, we get
$\mathcal{C}(i+1,0)=\sum_{l=0}^{k-1}\mathfrak{C}_\rho(i+1,l)Q(l,0)=\mathfrak{C}_\rho(i+1,0)Q(0,0)= c_{\rho^{-(i+1)}(0)}=c_{\rho^{k-i-1}(0)}= \mathcal{C}(i,k-1)$. This holds because $Q(0,0)=1$ and $Q(l,0)=0$ for $l=2,3,\hdots,k-1$. Similarly, by calculating the other entries of the $i+1$-th row, we get  $\mathcal{C}(i+1,j)=\sum_{l=0}^{k-1}\mathfrak{C}_\rho(i+1,l)Q(l,j)=\mathfrak{C}_\rho(i+1,\rho^j(0))Q(\rho^j(0),j)= c_{\rho^{-(i+1)}(\rho^j(0))}=c_{\rho^{-(i+1)+j}(0)}=c_{\rho^{k-(i+1)+j}(0)}=\mathcal{C}(i,j-1)$ for $j=1,2,\hdots,k-1$. Therefore $\mathcal{C}$ is a circulant matrix . 

To establish uniqueness, suppose there exists another permutation matrix $Q'$ defined using $k$-cycle $\rho'$ such that $\mathfrak{C}_{\rho'} Q'=$ circulant$(c_{0},c_{\rho(0)},c_{\rho^2(0)},c_{\rho^3(0)},\hdots,c_{\rho^{k-1}(0)})$. This implies $\rho^i(0)=\rho'^i(0)$ for $i=1,\hdots,k-1$. Since both $\rho$ and $\rho'$ are $k$-cycle permutation, we can conclude that $\rho=\rho'$ and $Q=Q'$.
\end{proof}

\begin{remark}
Note that, the Theorem \ref{relation cyclic circulant} has been discussed in Theorem $3$ of \cite{LS} with a unclear proof. Therefore we record here a rigorous proof for the sake of completeness.
\end{remark}

The matrix circulant$(c_{0},c_{\rho(0)},c_{\rho^2(0)},c_{\rho^3(0)},\hdots,c_{\rho^{k-1}(0)})$ is referred as the circulant matrix associated with the cyclic matrix $\mathfrak{C}_{\rho}$ throughout the paper. 

\begin{lemma}\label{structure}
Consider the permutation matrix $Q$ defined as $$Q(i,j)=\begin{cases}
  1 & \text{if $i=\rho^j(0),j=0,1,\cdots, k-1$; }\\
  0 & \text{otherwise}.
\end{cases}$$ Then $Q^{-1}=\text{cyclic}_\rho (1,0,0,\hdots,0)$.
\end{lemma}

\begin{proof}
Since $Q^{-1}=Q^{T},$ we can express $Q^{-1}$ as follows:  $$Q^{-1}(i,j)=\begin{cases}
  1 & \text{if $j=\rho^i(0),i=0,1,\cdots, k-1$; }\\
  0 & \text{otherwise}.
\end{cases}$$ 
Therefore $Q^{-1}$ has $1$ at the positions $\{(0,0),(1,\rho(0)), (2,\rho^2(0)), \hdots, (k-1, \rho^{k-1}(0))\}.$ Also from Definition \ref{cyclic def}, the cyclic matrix $\text{cyclic}_\rho(1,0,0,\hdots,0)$ has $1$ at positions $\{(0,0),(1,\rho(0)), (2,\rho^2(0)), \hdots, (k-1, \rho^{k-1}(0))\}$. Hence they are equal.
\end{proof}
Note that, the matrix $Q_\rho=\text{cyclic}_\rho (1,0,0,\hdots,0)$ satisfy $Q_\rho Q_\rho^T=I$.

An illustration of Theorem \ref{relation cyclic circulant} and Lemma \ref{structure} is the following.
\begin{example}
Consider the cyclic matrix $\mathfrak{C}_{\rho}(c_0,c_1,c_2,c_3,c_4)$ with $\rho=(0~2~3~1~4)$. Then $\mathfrak{C}_{\rho}=\begin{bmatrix}
c_0&c_1&c_2&c_3&c_4\\
c_4&c_3&c_0&c_2&c_1\\
c_1&c_2&c_4&c_0&c_3\\
c_3&c_0&c_1&c_4&c_2\\
c_2&c_4&c_3&c_1&c_1\\
\end{bmatrix}$. Using the construction of the permutation matrix given in Theorem \ref{relation cyclic circulant}, we have $Q=\begin{bmatrix}
1&0&0&0&0\\
0&0&0&1&0\\
0&1&0&0&0\\
0&0&1&0&0\\
0&0&0&0&1\\
\end{bmatrix}.$ Therefore $\mathfrak{C}_{\rho}Q$ is the circulant matrix with the first row $(c_0,c_2,c_3,c_1,c_4)$. Also, $Q^{-1}=\begin{bmatrix}
1&0&0&0&0\\
0&0&1&0&0\\
0&0&0&1&0\\
0&1&0&0&0\\
0&0&0&0&1\\
\end{bmatrix}= \text{cyclic}_{\rho}(1,0,0,0,0)$.
\end{example}

Next, we prove the generalization of Theorem \ref{structure g circulant}.
\begin{theorem}\label{structure cyclic}
Let $\mathfrak{C}_{\rho}(c_0,c_1,c_2,\hdots,c_{k-1})$ be a cyclic matrix.
 Then $\mathfrak{C}_{\rho}=\sum_{i=0}^{k-1}c_{\rho^{i}(0)}P^{i}Q_\rho$, where $Q_\rho= \text{cyclic}_\rho (1,0,0,\hdots,0)$ corresponding to the $k$-cycle $\rho$ and $P=$ circulant$(0,1,0,\hdots,0)$. 
\end{theorem}

\begin{proof}
Since $\mathfrak{C}_{\rho}$ is a cyclic matrix, applying Theorem \ref{relation cyclic circulant} we get a circulant matrix $\mathcal{C}=$ circulant$(c_{0},c_{\rho(0)},c_{\rho^2(0)},c_{\rho^3(0)},\hdots,c_{\rho^{k-1}(0)})$ corresponding to $\mathfrak{C}$. The circulant matrix $\mathcal{C}$ can be written as $\mathcal{C}=c_{0}I+c_{\rho(0)}P+c_{\rho^2(0)}P^2+\cdots+c_{\rho^{k-1}(0)}P^{k-1}$, where $P=$ circulant $(0,1,0,\cdots,0)$. Therefore, the cyclic matrix can be written as $\mathfrak{C}_{\rho}=c_{0}Q^{-1}+c_{\rho(0)}PQ^{-1}+c_{\rho^2(0)}P^2Q^{-1}+\cdots+c_{\rho^{k-1}(0)}P^{k-1}Q^{-1}$. By using Lemma \ref{structure}, we get $Q^{-1} = Q_\rho$. Hence proved. 
\end{proof}

\begin{example}\label{example 2}
Consider the matrix $\mathfrak{C}= \text{cyclic}_{\rho_1}(c_0,c_1,c_2,c_3,c_4)$ with  $\rho_1=(0~2~3~1~4)$ from Example \ref{example 1}. Then  $\mathfrak{C}=c_0Q_{\rho_1}+c_2PQ_{\rho_1}+c_3P^2Q_{\rho_1}+c_1P^3Q_{\rho_1}+c_4P^4Q_{\rho_1}$, with $Q_{\rho_1}$ is the matrix $Q_{\rho_1}=\begin{bmatrix}
1&0&0&0&0\\
0&0&1&0&0\\
0&0&0&1&0\\
0&1&0&0&0\\
0&0&0&0&1
\end{bmatrix}$.
\end{example}

Since $g$-circulant matrices with $\gcd(k,g)=1$ are cyclic matrices, we prove that Theorem \ref{structure cyclic} is essentially reduced to Theorem \ref{structure g circulant} for the $k$-cycle  $(0 \quad g \quad{2g\pmod k}\quad{3g\pmod k} \cdots {(k-1)g \pmod k})$ in the following corollary.

\begin{corollary}
Let $\mathfrak{C}_{\rho}=$ cyclic$(c_0,c_1,c_2,\hdots,c_{k-1})$, where $\rho$ is the $k$-cycle permutation 
$(0 ~ g ~{2g\pmod k}~ {3g\pmod k} \cdots {(k-1)g \pmod k})$ with $\gcd(g,k)=1$. Then $\mathfrak{C}_{\rho}=\sum_{i=0}^{k-1} c_iQ_gP^{i}$. 
\end{corollary}
\begin{proof}
Since $\mathfrak{C}_{\rho}=$ cyclic$(c_0,c_1,c_2,\hdots,c_{k-1})$, applying Theorem \ref{structure cyclic} we get 
\begin{equation*}
\mathfrak{C}_{\rho}=c_0Q_\rho+c_{\rho(0)}PQ_\rho+c_{\rho^2(0)}P^2Q_\rho+\cdots+c_{\rho^{k-1}(0)}P^{k-1}Q_\rho, 
\end{equation*}
where $Q_\rho= \text{cyclic}_\rho (1,0,0,\hdots,0)$. Note that $Q_\rho=Q_g$. Using Lemma \ref{PA=AP^g} and substituting $\rho(0)=g, \rho^i(0)=ig, 1 \leq i \leq k-1$ we get
\begin{equation*}
\mathfrak{C}_{\rho}=c_0Q_g+c_gQ_gP^g+c_{2g}Q_gP^{2g}+c_{3g}Q_gP^{3g}+\cdots+c_{(k-1)g}Q_gP^{(k-1)g}, 
\end{equation*}
where $\{ig, 1 \leq i \leq k-1\}$ are calculated modulo $k$. Using Lemma \ref{gcd result}, we can simply it further to
\begin{equation*}
\mathfrak{C}_{\rho}=c_0Q_g+c_1Q_gP+c_{2}Q_gP^{2}+\cdots+c_{(k-1)}Q_gP^{(k-1)}. 
\end{equation*}
This completes the proof.
\end{proof}

To determine the number of circulant matrices with same branch number, Liu and Sim \cite{LS} introduced an equivalence relation between two circulant matrices $\mathcal{C}=$ circulant$(c_0,c_1,\hdots,c_{k-1})$ and $\mathcal{C}_\sigma=$ circulant$(c_{\sigma(0)},c_{\sigma(1)},\hdots,c_{\sigma(k-1)})$. Their result is noted in the following theorem.
\begin{theorem}
Given two circulant matrices $\mathcal{C}$ and $\mathcal{C}_\sigma$, $\mathcal{C} \sim_{P.E} \mathcal{C}_\sigma $ if and only if $\sigma$ is some index permutation satisfying $\sigma(i)=bi+a \pmod k$, for all $i=0,1,\hdots,k-1, a,b \in \mathbb{Z}_k, ~\gcd(b,k)=1$.
\end{theorem}

In Theorem \ref{relation cyclic circulant}, we established a permutation equivalence between the cyclic matrix $\mathfrak{C}_{\rho}$ and the circulant matrix $\mathcal{C}=$ circulant $(c_{0},c_{\rho(0)},c_{\rho^2(0)},c_{\rho^3(0)},\hdots,c_{\rho^{k-1}(0)})$, where $\rho$ is a cycle of length $k$. Therefore these two matrices have same branch number. To determine when two cyclic matrices have the same branch number, we need to prove an equivalence relation between them. This is accomplished in the following theorem.

\begin{theorem}
Let $\mathfrak{C}_{\rho_1}$ and $\mathfrak{C}_{\rho_2}$ be two cyclic matrices with the first row $(c_0,c_1,\hdots, c_{k-1})$. Then $\mathfrak{C}_{\rho_1} \sim_{P.E} \mathfrak{C}_{\rho_2}$ if and only if their corresponding circulant matrices are permutation equivalent.
\end{theorem}
\begin{proof}
Let $\mathfrak{C}_{\rho_1}$ be permutation equivalent $\mathfrak{C}_{\rho_2}$. Then there exists permutation matrices $P_1, P_2$ such that $P_1\mathfrak{C}_{\rho_1}P_2=\mathfrak{C}_{\rho_2}$. From Theorem \ref{relation cyclic circulant}, we get two permutation matrices $Q_{\rho_1}$ and $Q_{\rho_2}$ such that $\mathfrak{C}_{\rho_1}Q_{\rho_1}=\mathcal{C}_1$ and $\mathfrak{C}_{\rho_2}Q_{\rho_2}=\mathcal{C}_2$. This implies $P_1 \mathcal{C}_1Q_{\rho_1}^{-1}P_2=\mathcal{C}_2Q_{\rho_2}^{-1}$. This can be written as $P_1 \mathcal{C}_1P_3=\mathcal{C}_2$ where $P_3=Q_{\rho_1}^{-1}P_2Q_{\rho_2}$. Since $P_3$ is also a permutation matrix, we get $ \mathcal{C}_1 \sim_{P.E}  \mathcal{C}_2$.

Conversely, let $\mathcal{C}_1 \sim_{P.E}  \mathcal{C}_2$. Then there exists permutation matrices $P_1, P_2$ such that $P_1\mathcal{C}_1P_2=\mathcal{C}_2$. Since $\mathcal{C}_1 $ and $\mathcal{C}_2$ corresponds to $\mathfrak{C}_{\rho_1}$ and $\mathfrak{C}_{\rho_2}$ respectively, we have $P_1\mathfrak{C}_{\rho_1}Q_{\rho_1}P_2=\mathfrak{C}_{\rho_2}Q_{\rho_2}$ . This implies $P_1\mathfrak{C}_{\rho_1}P_3=\mathfrak{C}_{\rho_2}$ where $P_3=Q_{\rho_1}P_2Q_{\rho_2}^{-1}$. 
\end{proof}

\begin{example}
Let $\mathfrak{C}_{\rho_1}=$ cyclic $(c_0,c_1,c_2,c_3,c_4)$ with $\rho_1= (0~ 2~ 4~ 1~ 3)$ and 

$\mathfrak{C}_{\rho_2}=$ cyclic $(c_0,c_1,c_2,c_3,c_4)$ with $\rho_2=(0~ 3~ 1~ 4~ 2)$.

Then 

$\mathfrak{C}_{\rho_1}=\begin{bmatrix}
c_0 & c_1 & c_2 & c_3 &c_4\\
c_3 &c_4 & c_0 & c_1 & c_2\\
c_1 & c_2 & c_3 &c_4 & c_0\\
c_4 & c_0 & c_1 & c_2 & c_3 \\
c_2 & c_3 & c_4 & c_0 & c_1
\end{bmatrix}$  and $\mathfrak{C}_{\rho_2}=\begin{bmatrix}
c_0 & c_1 & c_2 & c_3 &c_4\\
c_2 &c_2 & c_4 & c_0 & c_1\\
c_4 & c_0 & c_1 &c_2 & c_3\\
c_1 & c_2 & c_3 & c_4 & c_0 \\
c_3 & c_4 & c_0 & c_1 & c_2
\end{bmatrix}$. 

Here $P_1\mathfrak{C}_{\rho_1}P_2=\mathfrak{C}_{\rho_2}$ where $P_1=\begin{bmatrix}
1 & 0& 0 & 0 &0\\
0 & 0 & 0 & 0 & 1\\
0 & 0 & 0 & 1 & 0\\
0 & 0 & 1 & 0 & 0\\
0 & 1 & 0 & 0 & 0
\end{bmatrix}$ and $P_2=I_5$. 

Circulant matrices corresponding to  $\mathfrak{C}_{\rho_1}$ and $\mathfrak{C}_{\rho_2}$ are $\mathcal{C}_1=$ circulant $(c_0,c_2, c_4, c_1, c_3)$ and $\mathcal{C}_2= $ circulant $(c_0,c_3, c_1, c_4,c_2)$ respectively. Calculating $P_1\mathcal{C}_1P_3$ where $P_3=Q_{\rho_1}^{-1}Q_{\rho_2}$, where $Q_{\rho_1}$ and $Q_{\rho_2}$ are as defined in Theorem \ref{relation cyclic circulant}, we get $\mathcal{C}_2.$  This implies $\mathcal{C}_1 \sim_{P.E} \mathcal{C}_2.$ 

On the other way , $\mathcal{C}_1=\begin{bmatrix}
c_0 & c_2& c_4& c_1& c_3\\
c_3 & c_0 & c_2& c_4& c_1\\
c_1 & c_3 & c_0 & c_2& c_4\\
c_4 & c_1 & c_3 & c_0 & c_2\\
c_2 & c_4 & c_1 & c_3 & c_0
\end{bmatrix} $ and $\mathcal{C}_2=\begin{bmatrix}
c_0 & c_3& c_1& c_4& c_2\\
c_2 & c_0 & c_3& c_1& c_4\\
c_4 & c_2 & c_0 & c_3& c_1\\
c_1 & c_4 & c_2 & c_0 & c_3\\
c_3 & c_1 & c_4 & c_2 & c_0
\end{bmatrix}$. Consider $P_1=P_2=\begin{bmatrix}
1 & 0 & 0 & 0 & 0\\
0 & 0 & 0 & 0 & 1\\
0 & 0 & 0 & 1 & 0\\
0 & 0 & 1 & 0 & 0\\
0 & 1 & 0 & 0 & 0
\end{bmatrix}$ we get $P_1\mathcal{C}_1P_2=\mathcal{C}_2$. It is easy to check that $P_3=Q_{\rho_1}P_2Q_{\rho_2}^{-1}=I_5$. Thus $P_1\mathfrak{C}_{\rho_1}I_5=\mathfrak{C}_{\rho_2}$ and this implies $\mathfrak{C}_{\rho_1} \sim_{P.E} \mathfrak{C}_{\rho_2}$.

\end{example}

\section{\bf Results on orthogonal g-circulant matrices\label{sec:orthogonal_property} }

To investigate the orthogonal property of $2^d \times 2^d$ circulant MDS matrices over $\mathbb{F}_{2^m}$, Gupta and Ray \cite{GR} proved that $A^{2^d}$ is a scalar matrix, where $A$ is a circulant matrix. After that, Liu and Sim extended these to left-circulant case in \cite{LS}. Since a left-circulant matrix is symmetric, involutory property and orthogonal property of such a matrix are equivalent.
Note that orthogonal property is not always preserved under permutation equivalence. Therefore, in this section, we study $g$-circulant matrices, with orthogonal property. In Lemma \ref{A^2^d g circulant}, we prove the similar result for  $g$-circulant matrices of order $2^d \times 2^d$ over $\mathbb{F}_{2^m}$. For that, we first require the following lemma.

\begin{lemma}
If $g >1$ be an odd integer then $2^d$ divides $\frac{g^{2^d}-1}{g-1}$.
\end{lemma}
\begin{proof}
$\frac{g^{2^d}-1}{g-1}=(g+1)(g^2+1)(g^{2^2}+1)\cdots(g^{2^{d-1}}+1)$. Since 
$g$ is an odd number, $(g+1),(g^2+1),\cdots,(g^{2^{d-1}}+1)$ are all even numbers. Hence, each term is divisible by $2$ and there are $d$ such terms. Therefore $2^d$ divides the product.
\end{proof}

\begin{lemma}\label{A^2^d g circulant}
Let $A=g$-circulant$(c_0,c_1,c_2,\hdots,c_{2^d-1}), g>1$ with $\gcd(g,2^d)=1$ over the finite field $\mathbb{F}_{2^m}$. Then $A^{2^d} =(c_0^{2^d}+c_1^{2^d}+c_2^{2^d}+\cdots+c_{2^d-1}^{2^d})I$.
\end{lemma}

\begin{proof}
$A= g$-circulant$(c_0,c_1,c_2,\hdots,c_{2^d-1})$. Therefore by applying  Theorem \ref{structure g circulant}, we have 
\begin{equation*}
A=c_0Q_g+c_1Q_gP+c_2Q_gP^2+\cdots+c_{2^d-1}Q_gP^{2^d-1}.
\end{equation*}
Therefore 
\begin{align*}
A^{2^d}&=(c_0Q_g+c_1Q_gP+c_2Q_gP^2+\cdots+c_{2^d-1}Q_gP^{2^d-1})^{2^d}\\
&=c_0^{2^d}Q_g^{2^d}+c_1^{2^d}(Q_gP)^{2^d}+c_2^{2^d}(Q_gP^2)^{2^d}+\cdots+c_{k-1}^{2^d}(Q_gP^{k-1})^{2^d}.
\end{align*}
Since $P=$ circulant$(0,1,0,\hdots,0)$ is a $2^d \times 2^d$ matrix, we can easily see that $P^{2^d}=I$. Additionally, $Q_g= g$-circulant$(1,0,0,\hdots,0)$ is a $2^d \times 2^d$ matrix. By applying Lemma \ref{gh-circulant}, we can say $Q_g^{2^d}$ is a $g^{2^d}$-circulant matrix. Since $g^{\phi({2^d})}=g^{2^d-1} \equiv 1 \pmod {2^d}$, where $\phi$ is the Euler's phi function, we have $g^{2^d} \equiv 1 \pmod {2^d}$ by squaring both sides. This implies, $Q_g^{2^d}= I$. Therefore, 
\begin{align*}
A^{2^d}&=c_0^{2^d}Q_g^{2^d}+c_1^{2^d}Q_g^{2^d}P^{\frac{g^{2^d}-1}{g-1}}+c_2^{2^d}Q_g^{2^d}P^{\frac{2(g^{2^d}-1)}{g-1}}+c_3^{2^d}Q_g^{2^d}P^{\frac{3(g^{2^d}-1)}{g-1}}+\cdots \\&+c_{k-1}^{2^d} Q_g^{2^d}P^{\frac{(k-1)(g^{2^d}-1)}{g-1}}.
\end{align*} 
Since $\frac{n(g^{2^d}-1)}{g-1} \equiv 0 \pmod {2^d},$ we get $P^{\frac{n(g^{2^d}-1)}{g-1}}=I$. Therefore $A^{2^d}=(c_0^{2^d}+c_1^{2^d}+c_2^{2^d}+\cdots+c_{2^d-1}^{2^d})I$.
\end{proof}
Using this lemma we can say about determinant of $g$-circulant matrices.

\begin{lemma}\label{det g circulant}
Let $A=g$-circulant$(c_0,c_1,c_2,\hdots,c_{2^d-1})$ with $\gcd(g,2^d)=1$ over the finite field $\mathbb{F}_{2^m}$.  Then $\det (A) =(\sum_{i=0}^{2^d-1} c_i)^{2^d}.$
\end{lemma}

\begin{proof}
Let $A$ be $g$-circulant$(c_0,c_1,\hdots,c_{2^d-1})$ and  $\det A= \bigtriangleup$. Then $\bigtriangleup^{2^d}=(\det A)^{2^d}=\det(A^{2^d})$. From Lemma \ref{A^2^d g circulant} and Lemma $4$ of \cite{GR}, $A^{2^d}=(\sum_{i=0}^{2^d-1} c_i^{2^d})I$. So, $\bigtriangleup^{2^d}=(\sum_{i=0}^{2^d-1} c_i^{2^d})^{2^d}$. This implies, $\bigtriangleup= \sum_{i=0}^{2^d-1} c_i^{2^d}=(\sum_{i=0}^{2^d-1} c_i)^{2^d}$.
\end{proof}

\begin{theorem}{\label{2^d g circulant MDS}}
Let $A$ be a $2^d \times 2^d$ $g$-circulant orthogonal matrix over the finite field $\mathbb{F}_{2^m}$ with $\gcd(g,2^d)=1$, then $A$ is not an MDS matrix.
\end{theorem}

\begin{proof}
Let $A=g$-circulant$(c_0,c_1,\hdots ,c_{2^d-1}).$ Let the rows of $A$ are denoted as $R_0,R_1,\cdots, R_{2^{d}-1}$. Since $A$ is orthogonal, $R_0\cdot R_j=0$ for $j=1,2,\hdots,2^{d}-1$. Consider the product $R_0\cdot R_{j}=0$ for $j=\{(2k+1)g \pmod{2^d},k=0,1,\hdots,2^{d-2}-1\}$. These products lead to the following equations:
\begin{eqnarray*}
\sum_{i=0}^{2^d-1}c_ic_{i+1}=0, \sum_{i=0}^{2^d-1}c_ic_{i+3}=0, \sum_{i=0}^{2^d-1}c_ic_{i+5}=0, \cdots, \sum_{i=0}^{2^d-1}c_ic_{i+2^{d-1}-1}=0
\end{eqnarray*}
where suffixes are modulo $2^d$. Adding these equations yields 
\begin{eqnarray*}
(c_0+c_2+\cdots+ c_{2^d-2})(c_1+c_3+\cdots+ c_{2^d-1})=0
\end{eqnarray*}
Note that $g$-circulant$(c_0,c_2,\hdots,c_{2^d-2})$ and $g$-circulant$(c_1,c_3,\hdots,c_{2^d-1})$ are two $2^{d-1} \times 2^{d-1}$ submatrices of $A$. Therefore according to Lemma \ref{det g circulant}, either one of two sub-matrices is singular. Therefore $A$ is not MDS.
\end{proof}

In the case of order $4 \times 4$, there are only two distinct $g$-circulant matrices: the circulant and the left-circulant matrices. However, if the order is $2^d \times 2^d,d>2$, total $2^{d-1}$ different $g$-circulant matrices exists, corresponding to $g=1,3,5,\hdots, 2^{d}-1$. Notably, for $g=1$ we obtain a circulant matrix, and for this, our findings are reduced to the first three results presented in Section $3$ of \cite{GR}. On the other hand, for the case  $g=2^d-1$, resulting matrix is left circulant. In this scenario, Theorem \ref{2^d g circulant MDS} provides a more general proof of the results presented in Section $5.2$ of \cite{LS}. 

In the next section, we discuss about cyclic matrices with orthogonal property.

\section{\bf Results on orthogonal cyclic matrices\label{sec:orthogonal_cyclic} }

We begin this section with an alternative proof of theorem \ref{2^d g circulant MDS} using the fact that the permutation matrices are orthogonal. This proof holds for a more general class, i.e., for cyclic matrices. Although this proof is compact, it does not describe any properties of cyclic matrices like the determinant or scalar structure. But for the sake of completeness, we record the proof here.
\begin{theorem}{\label{2^d g circulant}}
Let $\mathfrak{C}$ be a $2^d \times 2^d$ cyclic orthogonal matrix over $\mathbb{F}_{2^m}$. Then $\mathfrak{C}$ is not an MDS matrix.
\end{theorem}

\begin{proof}
Let $\mathfrak{C}$ be a $2^d \times 2^d$ cyclic orthogonal MDS matrix over $\mathbb{F}_{2^m}$. Then $\mathfrak{C}\mathfrak{C}^T=I$. From Theorem \ref{relation cyclic circulant} there exists a permutation matrix $Q$ such that $\mathfrak{C}Q=\mathcal{C}$, where $\mathcal{C}$ is a circulant matrix. Using the fact $Q^{-1}=Q_\rho,$ we get $\mathfrak{C}=\mathcal{C}Q_\rho$. Since $Q_\rho$ is a permutation  matrix, it is orthogonal. Therefore $\mathfrak{C}\mathfrak{C}^T=\mathcal{C}Q_\rho(\mathcal{C}Q_\rho)^T=\mathcal{C}Q_\rho Q_{\rho}^T\mathcal{C}^T=\mathcal{C}\mathcal{C}^T=I$. Therefore $\mathcal{C}$ is a $2^d \times 2^d$ circulant orthogonal matrix. It is also MDS from Corollary \ref{per eq are MDS}. This leads to a contradiction.
\end{proof}

According to Remark $5$ in \cite{GR}, circulant orthogonal MDS matrices of orders other than $2^d \times 2^d$ exist over the finite field of characteristic $2$. For example, $\mathcal{C}_1=$ circulant$(a,1+a^2+a^3+a^4+a^6,a+a^2+a^3+a^4+a^6)$ and $\mathcal{C}_2=$ circulant$(1,1,a,1+a^2+a^3+a^5+a^6+a^7,a+a^5,a^2+a^3+a^6+a^7)$ are two examples of circulant MDS matrices with orthogonal property over the finite field $\mathbb{F}_{2^8}$ with the irreducible polynomial $x^8 + x^4 + x^3 + x + 1 $, where $a$ is a primitive element of the finite field.

Similarly cyclic MDS matrices with orthogonal property can be found for orders other than $2^d \times 2^d$. For instance, in the symmetric group $S_3$, only two cycles of order $3$ exist, which are $\rho_1=(0~1~2)$ and $\rho_2=(0~2~1)$. This cycle $\rho_1$ produce the circulant matrix $\mathcal{C}_1=$ circulant$(a,1+a^2+a^3+a^4+a^6,a+a^2+a^3+a^4+a^6)$ and the cycle $\rho_2$ yields the left-circulant matrix $\mathfrak{C}_{\rho_2}=$ left-circulant$(a,1+a^2+a^3+a^4+a^6,a+a^2+a^3+a^4+a^6)$. Additionally, $\mathfrak{C}_{\rho_2}$ is an orthogonal (involutory) MDS matrix.

In the symmetric group $S_6,$ consider the $6$-cycle $\rho=(0~2~4~3~5~1)$. Consider the cyclic matrix $\mathfrak{C}_{\rho}$ with first row $(1, a^2+a^3+a^6+a^7, 1, 1+a^2+a^3+a^5+a^6+a^7,a, a^5+a)$. Then the matrix is $\mathfrak{C}_{\rho}=\begin{bmatrix}
\substack{1} & \substack{a^2+a^3\\+a^6+a^7} & \substack{1} & \substack{1+a^2+a^3+\\a^5+a^6+a^7} & \substack{a} & \substack{a^5+a}\\
\substack{a^2+a^3\\+a^6+a^7} & \substack{a^5+a} & \substack{1}  & \substack{a} & \substack{1}  & \substack{1+a^2+a^3+\\a^5+a^6+a^7}\\
\substack{a^5+a} & \substack{1+a^2+a^3+\\a^5+a^6+a^7} & \substack{a^2+a^3\\+a^6+a^7} & \substack{1}  & \substack{1}  & \substack{a}\\
\substack{1+a^2+a^3+\\a^5+a^6+a^7} & \substack{a} & \substack{a^5+a} & \substack{1}  & \substack{a^2+a^3\\+a^6+a^7} & \substack{1} \\
\substack{a} & \substack{1}  & \substack{1+a^2+a^3+\\a^5+a^6+a^7} & \substack{a^2+a^3\\+a^6+a^7} & \substack{a^5+a} & \substack{1} \\
\substack{1}  & \substack{1}  & \substack{a} & \substack{a^5+a} & \substack{1+a^2+a^3+\\a^5+a^6+a^7} & \substack{a^2+a^3\\+a^6+a^7}
\end{bmatrix}$. 

According to Theorem \ref{relation cyclic circulant}, the circulant matrix corresponding to $\mathfrak{C}_{\rho}$ is $\mathcal{C}=$ circulant$(1,1,a,1+a^2+a^3+a^5+a^6+a^7,a+a^5,a^2+a^3+a^6+a^7)$. This matrix is orthogonal and MDS. Therefore $\mathfrak{C}_{\rho}$ is also orthogonal and MDS. 

The general result is presented in the following theorem and the proof is straightforward using the identity $PP^T=I$ for any permutation matrix $P$.

\begin{theorem}
Let $\mathfrak{C}$ be a cyclic matrix of order $k \times k,  k\neq 2^d$. Then $\mathfrak{C}$ is orthogonal MDS if and only if the corresponding circulant matrix is orthogonal MDS. 
\end{theorem}

\begin{table*}
\centering
\begin{tabular}{||c c c c ||} 
 \hline
 Type of MDS Matrix & Order & Existence & Reference \\ [0.5ex] 
 \hline\hline
 circulant, orthogonal  & $2^d \times 2^d$ & do not exist &  \cite{GR}  \\ 
 \hline
 circulant, orthogonal & other orders & may exist & \cite{GR}  \\
 \hline
 left- circulant, orthogonal &  $2^d \times 2^d$ & do not exist & \cite{LS}  \\
 \hline
g-circulant, orthogonal  & $2^d \times 2^d$ & do not exist & this paper \\
 \hline
 cyclic, orthogonal & $2^d \times 2^d$ & do not exist & this paper\\  
 \hline
 cyclic, orthogonal & other orders & exist & this paper \\  
 \hline
 \end{tabular}\vspace{5pt}

\caption{List of pre-existing results and new results over finite fields of characteristic $2$.
 \label{tab:FPS}}
\end{table*}

\section{\bf Conclusion}
In this article, we focus on the structure of cyclic matrices and their application in the construction of MDS matrices. Cyclic matrices are closely related to circulant matrices and therefore offer flexible and efficient hardware implementations.  This makes the study of cyclic matrices both practically relevant and theoretically significant.
Due to these foundational properties, we study the characteristics of cyclic matrices and present new results on orthogonal $g$-circulant matrices with the MDS property. Additionally, we provide explicit examples of cyclic MDS matrices that also possess the orthogonal property.

 While this article establishes a general structure for cyclic matrices and cyclic orthogonal MDS matrices, it primarily focuses on matrices of size $2^d \times 2^d,$ where $d$ is a positive integer. However, the study of cyclic matrices with orders other than powers of two remains an open and intriguing area of research. Additionally, investigating the generalization of orthogonal properties, as well as involutory and semi-involutory properties in cyclic matrices, presents another promising direction for future study.

\bibliographystyle{plain}

\end{document}